\documentclass[journal]{IEEEtran}

\usepackage{comment}
\usepackage{cite}
\usepackage{amsthm}

\theoremstyle{plain}
\newtheorem{theorem}{Theorem}

\usepackage[table, svgnames, dvipsnames]{xcolor}
\usepackage{lineno,hyperref}
\modulolinenumbers[5]
\usepackage{tabularx}
\usepackage[utf8]{inputenc}

\usepackage{amsmath}
\usepackage{booktabs}
\newtheorem{corollary}{Corollary}
\usepackage{amssymb}
\usepackage{algorithm}
\usepackage{algorithmic}
\usepackage{mathtools}

\usepackage{graphicx}
\usepackage{xspace}
\usepackage{etoolbox}
\usepackage{physics}

\usepackage{tikz}
\usetikzlibrary{calc,positioning}

\usepackage[capitalize,nameinlink]{cleveref}

\crefname{proposition}{Proposition}{Propositions}
\Crefname{proposition}{Proposition}{Propositions}

\hypersetup{
    colorlinks,%
    citecolor=green,%
    filecolor=black,%
    linkcolor=red,%
    urlcolor=blue
}

 \newcommand{\tenp}[1]{\cdot 10^{#1}}

\title{Closed form logical error rate approximations for surface codes}

\author{Shaked~Regev,
        Daniel~Dilley,
        and Ryan~Bennink%
\thanks{Approved for public release; distribution unlimited. The views, opinions, and/or findings expressed are those of the authors and should not be interpreted as representing the official views or policies of DARPA or the U.S. Government.}%
\thanks{This manuscript has been authored by UT-Battelle, LLC, under Contract DE-AC05-00OR22725 with the U.S. Department of Energy (DOE). The U.S. Government retains and the publisher, by accepting the article for publication, acknowledges that the U.S. Government retains a nonexclusive, paid-up, irrevocable, worldwide license to publish or reproduce the published form of this manuscript, or allow others to do so, for U.S. Government purposes. DOE will provide public access to these results of federally sponsored research in accordance with the DOE Public Access Plan: \url{https://www.energy.gov/doe-public-access-plan}.}%
\thanks{Shaked Regev and Ryan Bennink are with Oak Ridge National Laboratory, Oak Ridge, TN 37830 USA.}%
\thanks{Daniel Dilley is with Argonne National Laboratory, Lemont, IL 60439 USA.}%
}

\begin{document}
\maketitle

\begin{abstract}
We propose a novel method to calculate logical error rates in surface codes, assuming independent and identically distributed physical errors. These results fit well known scaling laws for studied physical error rates and code distances, but break down in some other regimes. Increasing the code distance still reduces logical error rates, but returns diminish quickly beyond a certain point. We show how to use our method to analyze hypothetical quantum computers with various configurations and select designs with lower error rates. Currently, this  requires expensive classical simulations of quantum decoders for various distances and physical error rates or inaccurate extrapolation from minimal experimental data. Instead, we use the symmetry of the problem to count the configurations that result in a logical error with our novel software. Given a physical error rate, we can deduce the probability of a logical error to provably good accuracy quickly. We include an analysis of measurement errors to allow for a more complete comparison of different surface code implementations. We extend this counting method to a particular locally-correlated error model and another globally-correlated error model.
\end{abstract}

\section{Introduction}
\label{sec:introduction}

The planar surface code~\cite{KITAEV20032} is one of the most widely studied and 
accessible QEC codes. Its topological nature and reliance on only nearest-neighbor interactions make it particularly well suited for many quantum architectures. It encodes logical qubits using a two-dimensional array of physical qubits. Stabilizer measurements detect $X$ or $Z$-errors.  We consider the unrotated and rotated variants of planar codes.

In an \textit{unrotated surface code} the physical qubits are arranged on a square lattice, with stabilizers defined over plaquettes to detect both errors. Logical operators act on the encoded qubit as a whole. The logical $X$ operator $X_L$ is implemented by applying $X$ operators along a vertical chain of qubits spanning the lattice. The logical $Z$ operator $Z_L$ corresponds to a horizontal chain of $Z$ operators. 

A \textit{rotated surface code}~\cite{Bombin2007, Horsman2012SurfaceCode} is obtained by rotating the planar lattice $45^\circ$. 
This reduces the number of physical qubits required for a given $d$ by a factor of almost $2$, while preserving QEC capabilities. 
In this geometry, logical operators traverse zig-zag minimum-length paths connecting the lattice boundaries. $X_L$ operators go from top to bottom, while $Z_L$ operators go from left to right. The rotated lattice imposes constraints on possible configurations of physical qubits along these paths. Each logical operator may have multiple minimum-length logical paths (\textbf{MLLP}s). So for a given $d$ rotated codes have higher error rates. Overall, the tradeoff is such that rotated codes require $\approx 75\%$ of qubits of  unrotated codes to achieve similar accuracy~\cite{Orourke2025}.

Surface code decoders process stabilizer measurement outcomes (syndromes) to infer most likely configurations of physical errors. Logical errors arise when different physical error configurations produce identical syndromes, causing the decoder to apply the wrong correction. This effect is particularly significant in rotated surface codes, where overlapping MLLPs may share syndromes, and distinct physical error patterns become indistinguishable to the decoder. 
For sufficiently low physical error rates $p < p_{th}$, where $p_{th}$ is the surface code threshold, the logical error rate decreases exponentially with $d$~\cite{Dennis_2002}. Note that typically $p$ refers to the probability of an error from one physical operation on a physical qubit, not a logical operation on a physical qubit. Most qubits require $\approx 6$ physical operations per logical operation per measurement cycle. 



To design a practical quantum computer and assess its reliability, one must currently simulate every combination of $p$ and $d$ of interest. This is prohibitively expensive, particularly for small $p$ and large $d$. An alternative is to simulate a few points and extrapolate the rest, but this is inaccurate in the same regime~\cite{Fowler_2012}. Neither method can accurately capture the behavior of large-scale quantum systems in a reasonable time. 

Our framework allows analysis of logical error rates in unrotated and rotated surface codes. We first systematically enumerate the error configurations along the MLLPs for any $d$. We can then combine this with any $p$ and can approximately compute the logical error rate $L$ efficiently. Our approach is provably accurate in regimes of interest and avoids resource-intensive simulations.
It gives a rigorous upper bound on the quality of QEC for any imaginable surface code configuration within a minute of serial run time.


We assume that (i) our decoder is classical and perfect, (ii) unless stated otherwise, physical errors are independent (in space and time) and identically distributed (\textbf{i.i.d}) Pauli errors, i.e., they are Markovian, and (iii) measurement errors take a certain form (see \cref{sec:Meas}). \cref{tab:notation} summarizes the notation used in the majority of the paper.  

\begin{table*}[t] 
\caption{\label{tab:notation} Notation, not including notation specific to \cref{sec:locCOR}, \cref{sec:global}, or \cref{sec:Meas} }

\centering
\smallskip
\footnotesize

\begin{tabular}{lll} \toprule
   Variable  & Meaning     & Function relations
\\ \midrule
$d$ & code distance & $d\%2=1$\\
$p$ & physical error rate from one logical operation & \\
$L$ & logical error rate & $L=f(p,d)$ - goal is to determine $f()$\\
$C_k$ & \# configurations of $k$ physical errors resulting in logical error & $C_k=g(d)$ - goal is to determine $g()$\\
$P_k$ & probability of $k$ physical errors resulting in logical error & $P_k=h(p,d)$ - goal is to determine $h()$\\
$d_e$ & minimum \# of physical errors required for a logical error & $d_e\doteq (d+1)/2$\\
$p_{th}$ & threshold below which $L$ decays exponentially in $p$ & $p_{th}>p$
\\ \bottomrule
\end{tabular}
\end{table*}

In \cref{sec:MLLP} we prove a limitation on the predicted scaling of logical error rates given $p,d$~\cite{Fowler_2012} to $p,d$ such that $pd^2\ll 1$. This limitation has not yet been shown experimentally, because $d$ is too small on real quantum hardware. We prove that $pd^2\ll 1$ is sufficient for the most likely configurations to be the only ones that contribute meaningfully to the logical error rate. In \cref{sec:cde}, we present a novel algorithm and software that efficiently and provably accurately calculate the number of most likely physical error configurations that cause logical errors for any $d$. Then, any $p$ can be entered to obtain the corresponding $L$s. 
In \cref{sec:locCOR}, we introduce a locally correlated physical error model and study its effects. In \cref{sec:global}, we extend our model to particular correlated global noise.
In  \cref{sec:Meas}, we show that taking a number of measurements equal to $d$ is sufficient to make measurement errors negligible for rotated surface codes.
\cref{sec:future} summarizes our work and proposes future research directions.

\section{The minimum-length logical path (\textbf{MLLP}) problem}
\label{sec:MLLP}
We consider a surface code with $d^2$ data qubits ($d$), where each qubit experiences an independent and identically distributed (i.i.d.) 
error with probability $p$
Our objective is to derive a provably accurate approximation for $L$.
We can use this form to calculate the error rates for logical qubits with small $p$. Additionally, we can modify it to model correlated errors or measurement errors.
The logical error rate is 
\begin{equation}
\label{eq:logicalError}
    L = \sum_{k=0}^{d^2} C_k p^k (1-p)^{d^2-k}\doteq  \sum_{k=0}^{d^2} P_k,
\end{equation}
where $C_k$ is the number of distinct physical error configurations involving $k$ qubits that result in a logical error. Logical failures occur only when errors collectively form a path that is equivalent to a logical operator. For a surface code of distance $d$, this requires at least $k\geq(d+1)/2 = d_e$ physical errors. Consequently, $\forall k<d_e, C_k=0$. 

If the physical noise model includes $X$ and $Z$-errors on each qubit (i.e., $Y$ errors up to a phase), then 
$C_{d_e}$ is exactly doubled by symmetry. Doubling does not introduce over-counting, because a given error path cannot simultaneously be an $X$-type and a $Z$-type logical path; the two classes of paths are topologically and syndrome-wise disjoint. Without loss of generality (\textbf{wlog}), we restrict attention to $X$-errors. 


 Any configuration with $d_e$ physical errors that creates a logical error can have errors arbitrarily added anywhere on the surface and will create a logical error if there are no error cancellations.
    There are at most $d^2-d_e$ locations for these next errors (exactly $d^2-d_e$ locations neglecting error cancellations). Any set of $k$ errors that did not create a logical error, but did when an error was added, is accounted for by switching this last error with one that was not in the error path. Therefore, $\forall k\geq 0$,  $C_{d_e+k}\leq \binom{d^2-d_e}{k}C_{d_e}\leq \frac{d^{2k}}{k!}C_{d_e}$. So $  \forall k\geq 0$

\begin{align}
P_{d_e+k}< \frac{\left(d^2p\right)^kP_{d_e}}{k!},
    \label{eq:pnext}
\end{align}
and $P_k$s decay exponentially if $pd^2< 1$. So when $pd^2\ll 1$, $L\approx P_{d_e}$. The inequality in \cref{eq:pnext} becomes looser as $k$ and the probability of error cancellation increases, making the approximation $L\approx P_{d_e}$ even better. So truncating \cref{eq:logicalError} at the first non-zero term is accurate in this regime, but perhaps not in the regime $pd^2\gtrsim 1$.

\begin{theorem}
    \label{thm:upplow}
   $\forall d\geq 3:$ $\sum_{k=0}^{d^2-3d_e}\binom{d^2-3d_e}{k}(\frac{p}{1-p})^kP_{d_e}\\\leq L\leq \sum_{k=0}^{d^2-2d_e}\binom{d^2-2d_e}{k}(\frac{p}{1-p})^kP_{d_e}$.
\end{theorem} 
\begin{proof}
Any set of $d_e+k$ physical errors can create a physical error only if (but not necessarily if) it has a subset of $d_e$ errors that create a logical error. For X errors, error `cancellation' occurs when there are 2 errors that are adjacent horizontally. The error decoder does not correctly discern the errors, but the correction it applies still restores the logical state if one of these errors was necessary to complete a set of $d_e$ errors on an MLLP.

There are at least $d^2-3d_e$ locations to choose $k$ additional errors that do not cancel out the original $d_e$ errors, so $\binom{d^2-3d_e}{k}C_{d_e}\leq C_{d_e+k}$. There are at most $d^2-2d_e$ locations for these errors, so $C_{d_e+k}\leq \binom{d^2-2d_e}{k}C_{d_e}$. Multiplying by the probability of a given configuration with $k$ more errors and summing over all $k$ gives the result.
\end{proof}
Note that the lower bound corresponds exactly to sets with no error on the edge and the upper  bound corresponds exactly to sets where all errors are on the edge.
\begin{corollary}
If $pd\ll 1$, calculating $P_{d_e}$ is sufficient to determine $L$.
\end{corollary}
\begin{proof}

    $\sum_{k=0}^{n} \binom{n}{k}a^k=(1+a)^n$. Substituting \\$a=p/(1-p)$ and $n$ for the number of qubits gives a  simplification to \cref{thm:upplow}:
\begin{align}
\label{eq:Lboundtight}
        (1-p)^{3d_e-d^2}P_{d_e}\leq L \leq (1-p)^{2d_e-d^2}P_{d_e}.
\end{align}

\cref{eq:Lboundtight} is tight when $pd\ll 1$.
\end{proof}
\cref{fig:correction} shows $L/P_{d_e}-1$. As this factor increases, $L\approx P_{d_e}$ becomes inaccurate. This shows that truncating \cref{eq:logicalError} can be very inaccurate, even though $L$ can still be exponentially suppressed by increasing $d$.

 \begin{figure}[htbp]  
\centering
\includegraphics[width=\columnwidth]{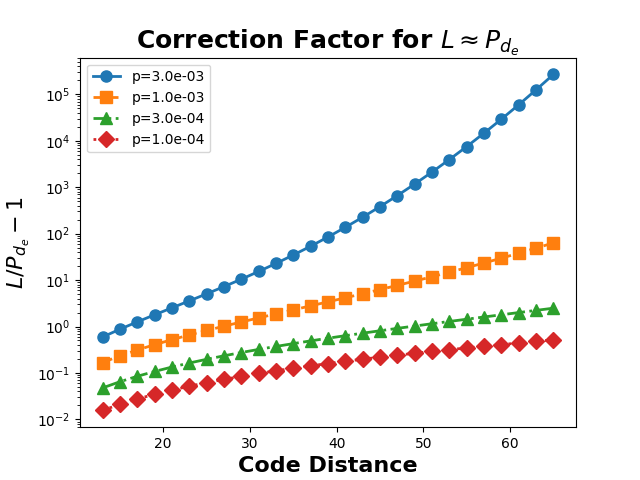}
  \caption{The correction factor increases with the code distance $d$ and physical error rate $p$. Truncating \cref{eq:logicalError} can be orders of magnitude off for large $p$ and $d$.}
  \label{fig:correction}
\end{figure}

\section{Calculating $C_{d_e}$ for surface codes}
\label{sec:cde}
\subsection{Unrotated codes}
\label{sec:unrotateMLLP}
The only subsets of $d$ physical qubits that fail to support $d_e$ errors without inducing a logical error are columns (for $X$-type errors) and rows (for $Z$-type errors). 
A logical error results if in any of the $d$ columns, $d_e$ of the qubits have errors~\cite{Fowler_2012}. 
We add a simplification to this formula (see proof in Appendix \ref{sec:AppendixA}):
\begin{align}
    \label{eq:Ccol}
    C_{d_e} = d\binom{d}{d_e} \approx \sqrt{\frac{d_e}{\pi}}4^{d_e}.
\end{align}
\label{sec:unrotate}
\label{sec:rotate}
For rotated codes, $C_{d_e}$ in \cref{eq:logicalError} is the number of ways $d_e$ physical errors can occur along an MLLP that traverses the lattice horizontally (for $Z$-errors) or vertically (for $X$-errors). For example, in \cref{fig:rotated} $d=5$, so any three errors on an MLLP will produce a logical fault. We will therefore count the number of such error configurations and use it to approximate $L$. \cref{fig:syndrome} illustrates how a configuration with three errors can be mistaken for the more likely configuration containing only two errors on the same chain. $X$ ($Z$)-error chains pass only through green (pink) squares.
\begin{figure}[hptb]   %
\centering
\includegraphics[width=\columnwidth]{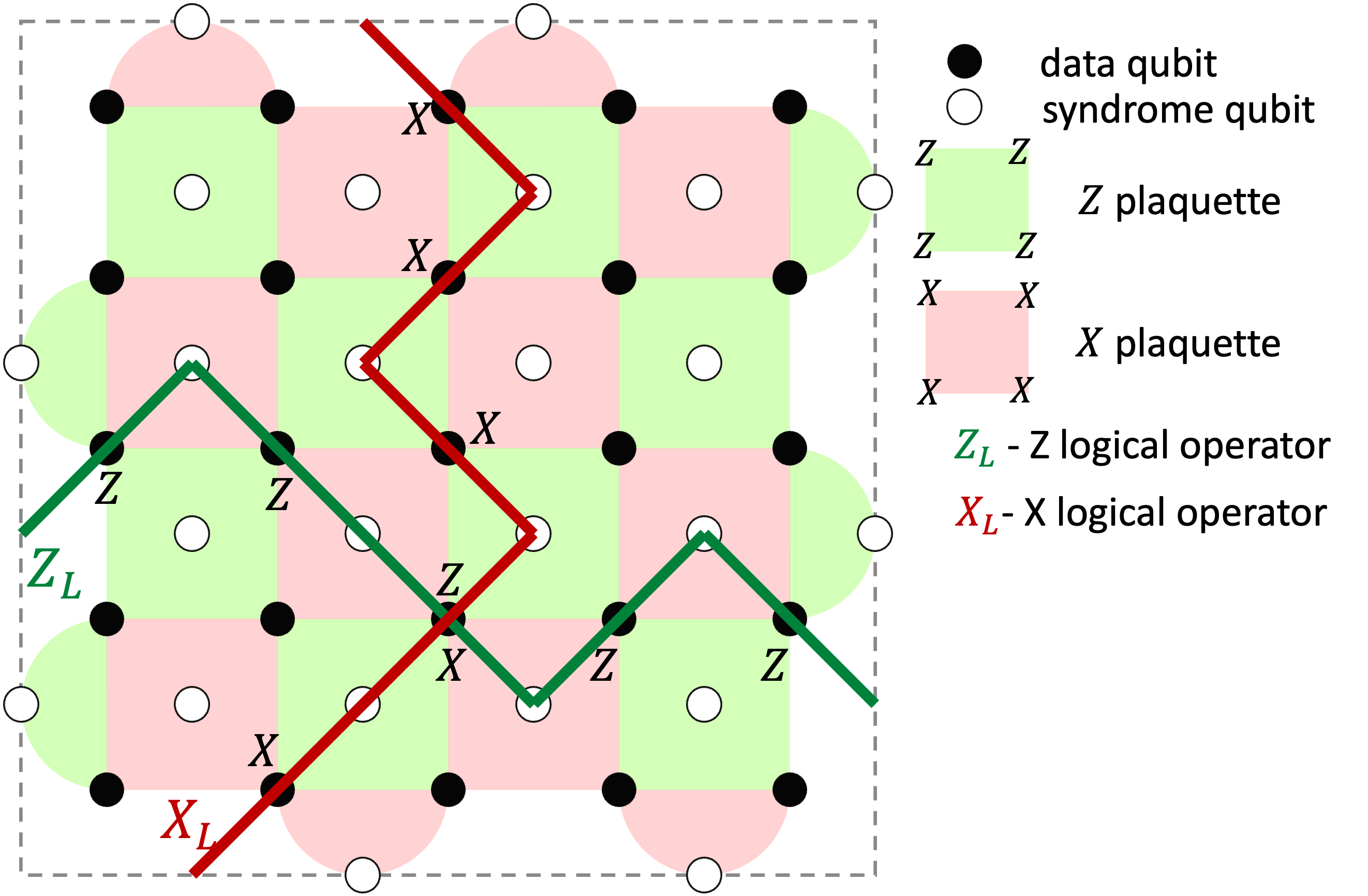}
  \caption{Rotated surface code 
  that encodes a logical qubit using 25 data qubits and 24 syndrome qubits. 
  $X_L$ operators must pass through $Z$ stabilizer plaquettes and $Z_L$ operators must pass through $X$ stabilizer plaquettes. An $X$($Z$) error on a data qubit anticommutes with the neighboring $Z$($X$)-type  stabilizers and flips the parity (i.e., odd or even) of their measurement outcomes. An additional error on a data qubit adjacent to the same plaquette flips the parity of the syndrome qubit again, causing the error to go unnoticed (see \cref{fig:syndrome}.)}
\label{fig:rotated}
\end{figure}
\begin{figure}[hptb]   %
\centering
\includegraphics[width=\columnwidth]{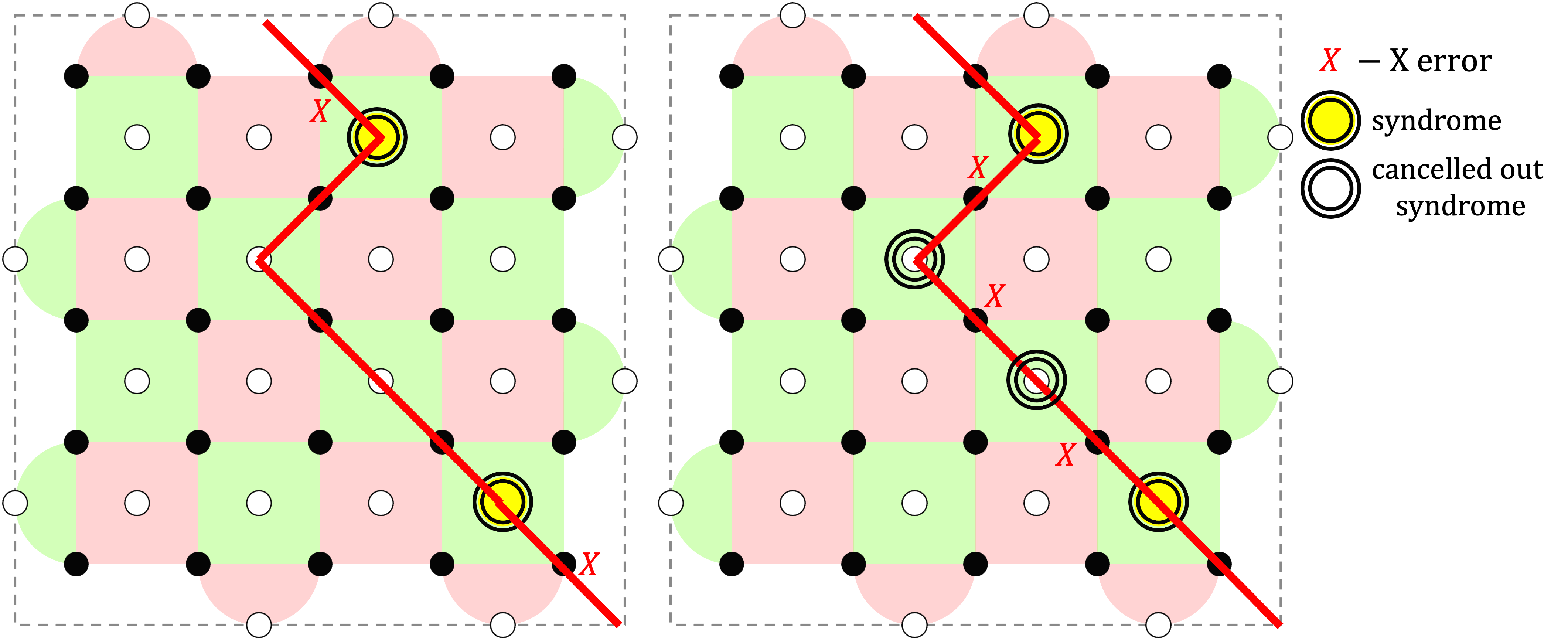}
  \caption{Configurations of 2 and 3 errors that the decoder cannot distinguish. 2 errors on the same plaquette cancel out the syndrome. The left configuration is more likely because it has fewer errors. However, the decoder cannot distinguish these two configurations. So, if the right configuration occurs, the decoder will incorrectly identify it as the left configuration and apply the wrong correction. This results in a logical error.}
\label{fig:syndrome}
\end{figure}

We build on the approach in~\cite{Aliferis2008} of counting paths to estimate the logical error rate. Focusing on $X$-errors only, we obtain the following upper bound on $C_{d_e}$:
\begin{equation}
    C_{d_e}\leq d2^{d-1}\times \binom{d}{d_e}=\frac{1}{4}2^{d+1}\times d\binom{d}{d_e}\leq \sqrt{\frac{d_e}{16\pi}}16^{d_e}.
    \label{eq:rotatedUbound}
\end{equation}
$d2^{d-1}$ is an upper bound on the number of MLLPs. The path may start at any of the $d$ boundary points and may proceed in at most two different ways at each of the subsequent $d-1$ steps. 
This estimate is asymptotically tight because, as $d$ grows, an increasing fraction of these paths remain entirely in the interior of the code and thus experience no boundary-induced constraints. 
\cref{fig:normalizedpaths} demonstrates the accuracy of $d2^{d-1}$ as the approximation of the number of paths, denoted by $N_{paths}$\footnote{Note that this shows that when $d$ is very large, $L$ is much closer to its lower bound in \cref{eq:Lboundtight}. So when $pd\gtrsim 1$, $L\approx (1-p)^{3d_e-d^2}P_{d_e}$ is a good approximation.}.
\begin{figure}[htpb]   
\centering
  \includegraphics[width=\columnwidth]{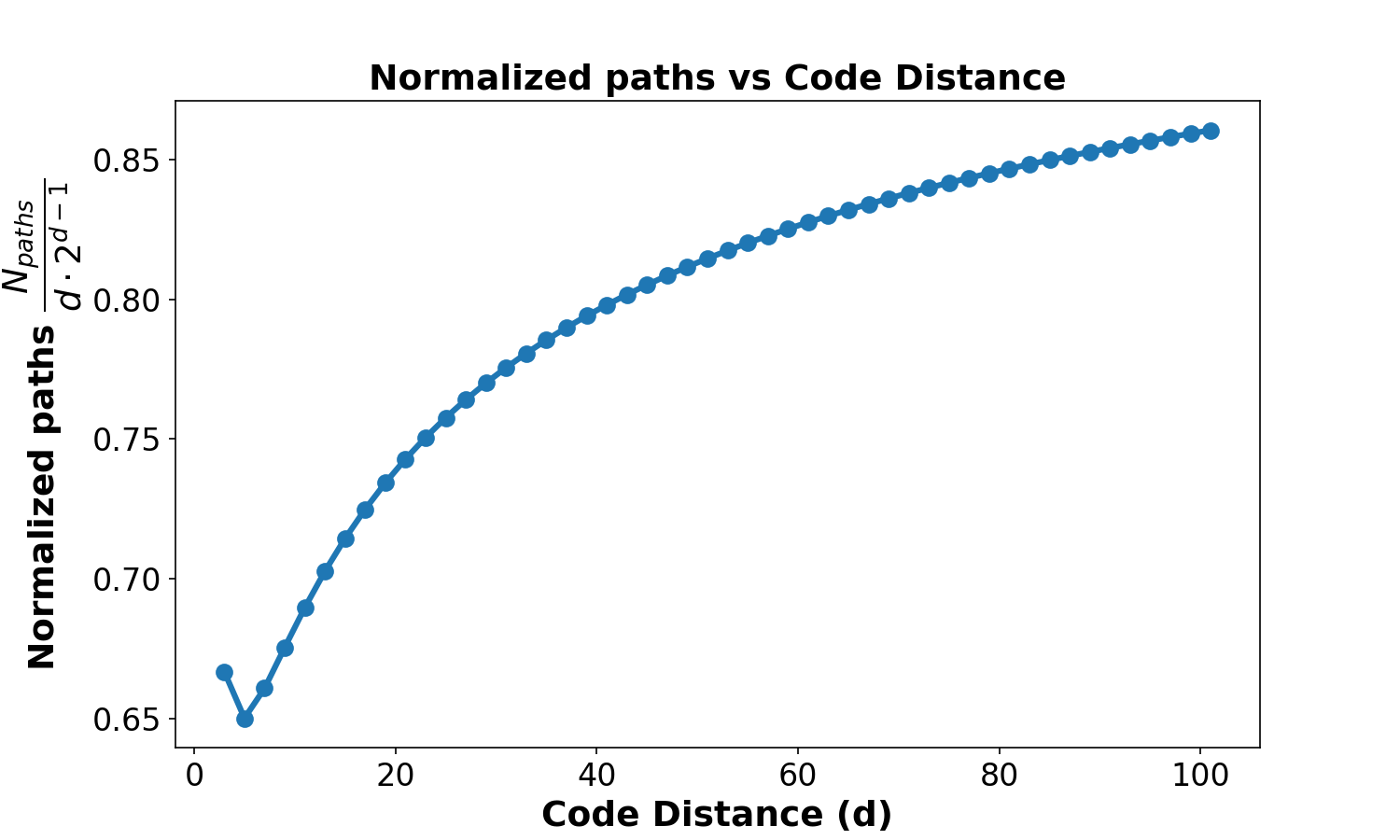}
  \caption{Approximating the number of paths $N_{paths}$ as $d2^{d-1}$ becomes increasingly tight asymptotically.}
  \label{fig:normalizedpaths}
\end{figure}

The second term in ~\cref{eq:rotatedUbound} counts the number of ways to arrange errors along an MLLP. This approximation is a union bound because a fixed set of physical error locations may lie on multiple distinct MLLPs, causing such configurations to be counted more than once. 
As \cref{fig:normalizedpaths} indicates, any asymptotic divergence between the upper bound and the true value of $C_{d_e}$ must originate entirely from this union-bound. 

Similarly, we obtain the following lower bound to $C_{d_e}$:
\begin{equation}
    C_{d_e}\geq \frac{Cd2^{d-1}\times \binom{d}{d_e}}{2^{d_e-1}}=\frac{C2^{d_e}}{2}\times d\binom{d}{d_e}\geq \sqrt{\frac{d_e}{16\pi}}8^{d_e}.
    \label{eq:rotatedLbound}
\end{equation}
The main difference from \cref{eq:rotatedUbound} arises from dividing by $2^{d_e-1}$, the maximum number of MLLPs that can contain the same set of physical errors. The constant $C>0.6$ reflects the edge effects, which become increasingly negligible as $d$ increases (see \cref{fig:normalizedpaths}). The gap between the bounds is large, so we turn to the precise geometric conditions under which a logical error arises. 
\begin{theorem}
\label{thm:firstcoefficient}
   Precisely $d_e$ physical errors produce a logical error if and only if one of two conditions holds for every pair of  physical errors that is consecutive vertically:
   \begin{itemize}
       \item They are closer horizontally than they are vertically.
       \item They are diagonal on the grid (equidistant vertically and horizontally) and the diagonal path connecting them passes only through  green plaquette(s). 
   \end{itemize}
\end{theorem}
\begin{proof}
    True by construction of MLLPs on rotated codes~\cite{Bombin2007,Horsman2012SurfaceCode, Orourke2025}. See also \cref{fig:syndrome}.
\end{proof}
The code we provide in \href{https://github.com/ORNL/QuantumChains}{this Github repository} computes $C_{d_e}$ exactly and independently of $p$. \cref{alg:errorPatterns} provides a simplified overview of this method. 
\begin{algorithm*}[htbp]
\caption{Count Error Patterns. The initial call is with $n=d_e$. See code \href{https://github.com/ORNL/QuantumChains}{here}}
\begin{algorithmic}[1]
\label{alg:errorPatterns}
\STATE \textbf{function} IsValidTransition($prev\_row, prev\_col, row, col$)
\IF{$row \leq prev\_row$}
    \STATE \textbf{return} false
\ENDIF
\STATE $dx \gets col - prev\_col$, $dy \gets row - prev\_row$
\IF{$dy < |dx|$}
    \STATE \textbf{return} false
\ENDIF
\IF{$dy = |dx|$  } 
\IF{$(row+col)\% 2 ==0$ \textbf{and} $dx < 0$}
    \STATE \textbf{return} false
\ENDIF
\IF{$(row+col)\% 2 ==1$ \textbf{and} $dx > 0$}
    \STATE \textbf{return} false
\ENDIF
\ENDIF
\STATE \textbf{return} true
\STATE
\STATE \textbf{function} CountMLLP($d, curr\_row, curr\_col, rem\_errs$)
\IF{$rem\_errs = 0$}
    \STATE \textbf{return} 1
\ENDIF
\IF{state is memoized}
    \STATE \textbf{return} memoized value
\ENDIF
\STATE $count \gets 0$
\FOR{$next\_row = curr\_row + 1$ \textbf{to} $d - 1$}
    \FOR{$next\_col = 0$ \textbf{to} $d - 1$}
        \IF{IsValidTransition($curr\_row, curr\_col, next\_row, next\_col$)}
            \STATE $count \gets count +$ CountMLLP($d, next\_row, next\_col, rem\_errs - 1$)
        \ENDIF
    \ENDFOR
\ENDFOR
\STATE memoize $count$
\STATE \textbf{return} $count$
\STATE
\STATE \textbf{function} CountErrorPatterns($d, n$)
\STATE $total \gets 0$
\FOR{$start\_row = 0$ \textbf{to} $d - n$}
    \FOR{$start\_col = 0$ \textbf{to} $d - 1$}
        \STATE $total \gets total +$ CountMLLP($d, start\_row, start\_col, n - 1$)
    \ENDFOR
\ENDFOR
\STATE \textbf{return} $total$
\end{algorithmic}
\end{algorithm*}
If $P_{d_e}$ dominates \cref{eq:logicalError}, our numerical results agree with the well known scaling law from Monte Carlo quantum error simulations~\cite{Fowler_2012}
\begin{equation}
    \label{eq:scaling}
    L\approx A\left(\frac{p}{p_{th}}\right)^{d_e},
\end{equation}

with $A\approx2.09\tenp{-1}$ and $p_{th} \approx7.33\tenp{-2}\doteq p_{t-r}$. \cref{fig:scalingfit} shows this agreement for $p=10^{-4}$. Recall that our definition of $p$ matches the literature for $p\leftarrow 6p$. Equivalently, $p$ can remain the same and $p_{t-r}\leftarrow p_{t-r}/6$. So $p_{t-r}\approx 1.22\tenp{-2}$, matching the estimated threshold values of $0.5\%-1\%$ closely. The small discrepancy comes from our assumptions of idealized decoding and measurements, which do not hold in practice.
\begin{figure}[htbp]   
\centering
  \includegraphics[width=\columnwidth]{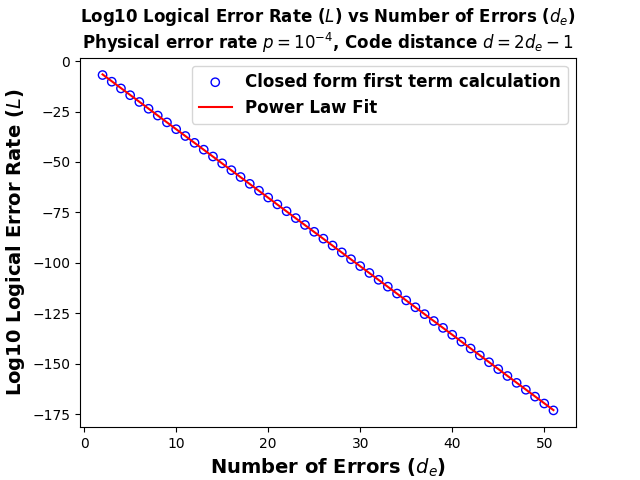}
  \caption{The code agrees with \cref{eq:scaling} with $A\approx2.09\tenp{-1}$ and $p_{th} \approx7.33\tenp{-2}$ ($R^2=1-6.32\tenp{-6}$).}
  \label{fig:scalingfit}
\end{figure}

Our upper bound given in \cref{eq:rotatedUbound} corresponds to $A\approx4.04\tenp{-1}$ and $p_{th} \approx6.22\tenp{-2}$. 
Our lower bound in~\cref{eq:rotatedLbound} corresponds to $A\approx4.04\tenp{-1}$ and $p_{th} \approx1.24\tenp{-1}$. 
The actual count is substantially closer to the upper bound than the lower bound.  

We fit \cref{eq:logicalError} to \cref{eq:scaling}, assuming that $P_{d_e}$ dominates the sum. We calculate $P_{d_e}$ by plugging $C_{d_e}$ from \cref{eq:Ccol} into \cref{eq:logicalError}. We get $A\approx 1.62$ and $p_{th}\approx 2.49\tenp{-1}\doteq p_{t-ur}$ for the unrotated code. For a given distance, this is substantially better than the rotated code, but the rotated code can have a distance larger by a factor of $\approx\sqrt{2}$. We compare the two and see that if $p^{(\sqrt{2}-1)d_e}\lesssim 7.75p_{t-r}^{\sqrt{2}d_e}/p_{t-ur}^{d_e}$, the rotated code has lower logical error rates per qubit count. This gives a closed form for the experimental results from~\cite{Orourke2025} and shows the limits of their extrapolation in certain regimes.  

Note that our analysis in \cref{sec:MLLP} did not depend on whether the code was unrotated or rotated. So $L\approx P_{d_e}$ if $pd^2\ll 1$, but the approximation becomes increasingly loose as $pd^2\gtrsim 1$. Combined, \cref{alg:errorPatterns} and  \cref{eq:Lboundtight} give a provably accurate estimate for $L$, even when $pd^2\gtrsim 1$. So, they are a computationally tractable version of \cref{eq:scaling} that matches it where it is correct and a more accurate version where it is not.


\section{A tractable, locally correlated model}
\label{sec:locCOR}
Suppose that an error on one qubit increases the probability of an error on neighboring qubits from $p$ to $ap$, with $a\geq1$. $a=1$ corresponds to the completely uncorrelated model. $a>1$ is plausible because local disturbances can affect a few adjacent qubits, or an error itself propagates via entanglement. 
If we were to take measurements sequentially, the probability of an error on the current qubit is $p$ if it is measured first or if the previous qubit had an error and $ap$ if the previous qubit had an error. This is unrealistic because errors are decoded simultaneously. 

We therefore symmetrize the resulting probability distribution by averaging probabilities for sequences with the same number of clusters of errors $q$. A cluster of errors is any non-empty sequence of $X$s (wlog $Z$s) on an MLLP. \cref{tab:3qubitrep} shows an example for the symmetrizing process on a 3 qubit repetition code. Note that $XXI$ and $IXX$ have one error cluster and are considered equally probable, while $XIX$ is considered less probable, because it has two clusters of errors.

\begin{table*}[htbp] 
\caption{\label{tab:3qubitrep} Different error possibilities and their probabilities for a 3 qubit repetition code with sequential measuring (left first) or symmetrized simultaneous measurement. $X$ denotes an error and $I$ denotes no error. $\rightarrow$ denotes probabilities that change in the symmetrized version. }
\centering
\smallskip
\footnotesize

\begin{tabular}{c|c|c} \toprule
   outcome & sequential measurement & symmetrized simultaneous measurement
\\ \midrule
XXX & $a^2p^3$ & $a^2p^3$\\
XIX & $p^2(1-ap)$&$p^2(1-ap)$\\
XXI & $ap^2(1-ap)$&$\rightarrow ap^2(1-p(a+1)/2)$ \\
IXX & $ap^2(1-p)$& $\rightarrow ap^2(1-p(a+1)/2)$ \\
IIX & $p(1-p)^2$&  $\rightarrow p(1-p)(1-p(2a+1)/3)$\\
IXI & $p(1-p)(1-ap)$&$\rightarrow p(1-p)(1-p(2a+1)/3)$ \\
XII & $p(1-p)(1-ap)$ &$\rightarrow p(1-p)(1-p(2a+1)/3)$ \\
III & $(1-p)^3$&  $(1-p)^3$
\\ \bottomrule
\end{tabular}
\end{table*} 

 In a distance $d$ surface code, assuming only errors on the same MLLP can propagate and cause correlated errors, the probability of a particular error configuration with $k$ errors in $q$ clusters is $p^{k}a^{k-q}\psi(a,p,k,q)$, with $(1-ap)^{d^2-k}\leq \psi(a,p,k,q) \leq (1-p)^{d^2-k}$. If $apd^2\ll 1$, which we expect is true for any useful quantum error correcting code, $\psi(a,p,k,q) \approx 1$. Otherwise, one can calculate $\psi$ via enumeration and averaging with the relevant variables, similarly to \cref{tab:3qubitrep}. \cref{eq:correlatedLogicalError} gives the logical error in the new correlated error model. \begin{equation}
\label{eq:correlatedLogicalError}
    L_a = \sum_{k,q} C_{k,q} p^k a^{k-q}\psi(a,p,k,q) \approx \sum_{k,q} C_{k,q} p^k a^{k-q},
\end{equation}
where $C_{k,q}$ is the number of distinct physical error configurations involving $k$ qubits in $q$ clusters that result in a logical error. Consequently, $\forall q,\;k<d_e\;:C_{k,q}=0$. Defining $C_k\doteq \sum_c C_{k,q}$ \cref{eq:pnext} still applies and the $C_{d_e}$ terms combined will dominate \cref{eq:correlatedLogicalError} if $apd^2< 1$. So $L_a \approx \sum_{q} C_{d_e,q} p^{d_e} a^{d_e-q}=p^{d_e}\sum_{q} C_{d_e,q} a^{d_e-q}$. In the following sections we will calculate $C_{d_e,q}$ for unrotated and rotated codes and calculate $L_a$ approximately when $apd^2 < 1$.

\subsection{Unrotated codes}
\label{sec:corrUnrotated}
We look at $X$ errors which happen only if there are at least $d_e$ errors in any column.
\begin{theorem}
    \label{thm:cdeq}
    For unrotated codes $C_{d_e,q}=d\binom{d_e-1}{q-1}\binom{d_e}{q}$.
\end{theorem}
\begin{proof}
    For any one of the $d$ columns, we are splitting $d_e$ errors into $q$ clusters. We are placing $q-1$ clusters of non-errors between the in the $d_e-1$ spaces available among the $d_e$ heads. This gives $\binom{d_e-1}{q-1}$ combinations. There are $d_e$ slots for the non-errors (including the beginning and the end) and we are placing $q$ error clusters among them. This gives   $\binom{d_e}{q}$ combinations. Multiplying these factors gives $C_{d_e,q}=d\binom{d_e-1}{q-1}\binom{d_e}{q}$.
\end{proof}
Vandermonde's identity provides a sanity check by comparing to \cref{eq:Ccol}
\begin{align*}
    &C_{d_e}=\sum_{q=1}^{d_e} C_{d_e,q}=d\sum_{q=1}^{d_e} \binom{d_e-1}{q-1}\binom{d_e}{q}\\=
    &d\sum_{q=1}^{d_e} \binom{d_e-1}{d_e-q}\binom{d_e}{q}=d\binom{d}{d_e}.
\end{align*}

Plugging this into our approximation for \cref{eq:correlatedLogicalError} gives
\begin{align}
\label{eq:correlatedLogicalErrorunrotated}
L_a \approx p^{d_e}d\sum_{q=1}^{d_e} \binom{d_e-1}{q-1}\binom{d_e}{q} a^{d_e-q}\nonumber\\=
p^{d_e}d\sum_{j=0}^{d_e-1} \overbrace{\left(1-\frac{j}{d_e}\right)\binom{d_e}{j}^2 a^{j}}^{K(j,a,d_e)}.
\end{align}
For given $p$, $d$, and $a$, we can now approximate $L_a$ to provable accuracy when $apd^2\ll 1$. The terms in the sum peak in the middle for $a=1$ and closer to the end for $a>1$. We will approximate the sum in \cref{eq:correlatedLogicalErrorunrotated} to get an idea of how it scales. We first find $K_{\max}(a,d_e) \doteq \max_j K(j,a,d_e)$. By Stirling's approximation, 
\begin{align*}
   &ln{K(j,a,d_e)}\approx\ln{(1-j/d_e)}+\\
   &2\left(d_e\ln{d_e}-(d_e-j)\ln{(d_e-j)}-j\ln{j}\right)+j\ln{a}. 
\end{align*}

Dividing by $d_e$ and changing variables $x=j/d_e$ ($0\leq x < 1$) gives
\[
\frac{\ln{K}}{d_e}=\frac{\ln{(1-x)}}{d_e}\overbrace{-2\left(x\ln{x}+(1-x)\ln{(1-x)}\right)+x\ln{a}}^{f(x)}.
\]
The first term is negligible for large $d_e$. We set the derivative of $f(x)$ to $0$
\[
f'(x)=2\ln{\left(\frac{1-x}{x}\right)}+\ln{a}=0. 
\]
$f(x)$ is therefore maximized at $x_0(a) = \sqrt{a}/(1+\sqrt{a})$. $x_0(1)=0.5$ confirms these approximations are reasonable because the binomial is maximized in the middle. Plugging this back into $K_{\max}(a,d_e)$ and using Stirling's approximation gives
\begin{equation}
    \label{eq:kmax}
    K_{\max}(a,d_e) \approx \frac{e^{d_ef(x)}}{2\pi d_e x} = \frac{(\sqrt{a}+1)^{2d_e+1}}{2\pi d_e \sqrt{a}}.
\end{equation}
If $d_e$ is large, the sum in \cref{eq:correlatedLogicalErrorunrotated} is approximately Gaussian and we can use Laplace's approximation if we have the width of the distribution. The width of the distribution is $\sqrt{\frac{2\pi d_e}{\abs{f''(x_0)}}}$. 
Using $
f''(x_0)=\frac{2}{x(x-1)}|_{x=\frac{\sqrt{a}}{\sqrt{a}+1}}=-\frac{2(\sqrt{a}+1)^2}{\sqrt{a}}
$ and \cref{eq:kmax} in \cref{eq:correlatedLogicalErrorunrotated} we get 
\begin{equation*}
L_a \approx p^{d_e}\frac{d(\sqrt{a}+1)^{2d_e}}{2\sqrt{\pi d_e}a^{1/4}}\lesssim L_a \approx p^{d_e}\sqrt{\frac{d_e}{\pi}}\frac{(\sqrt{a}+1)^{2d_e}}{a^{1/4}}.
\end{equation*}
For $a=1$ this matches the result in \cref{sec:unrotateMLLP}. Using \cref{eq:scaling} we can interpret this result as $p_{th}\approx (\sqrt{a}+1)^{-2}$. This shows that even a moderate $a$ such as $a\geq 2$ can results in order(s) of magnitude more frequent logical errors. Further, any  practical quantum error correcting code should have $a\lesssim 10$. 

\subsection{Rotated codes}
\label{sec:rotateCorr}
Expanding \cref{alg:errorPatterns} to distinguish between configurations that have the same number of errors, but are in different clusters, requires merely keeping track of the number of clusters and for every configuration. 

 \begin{figure}[thbp]  
\centering
\includegraphics[width=\columnwidth]{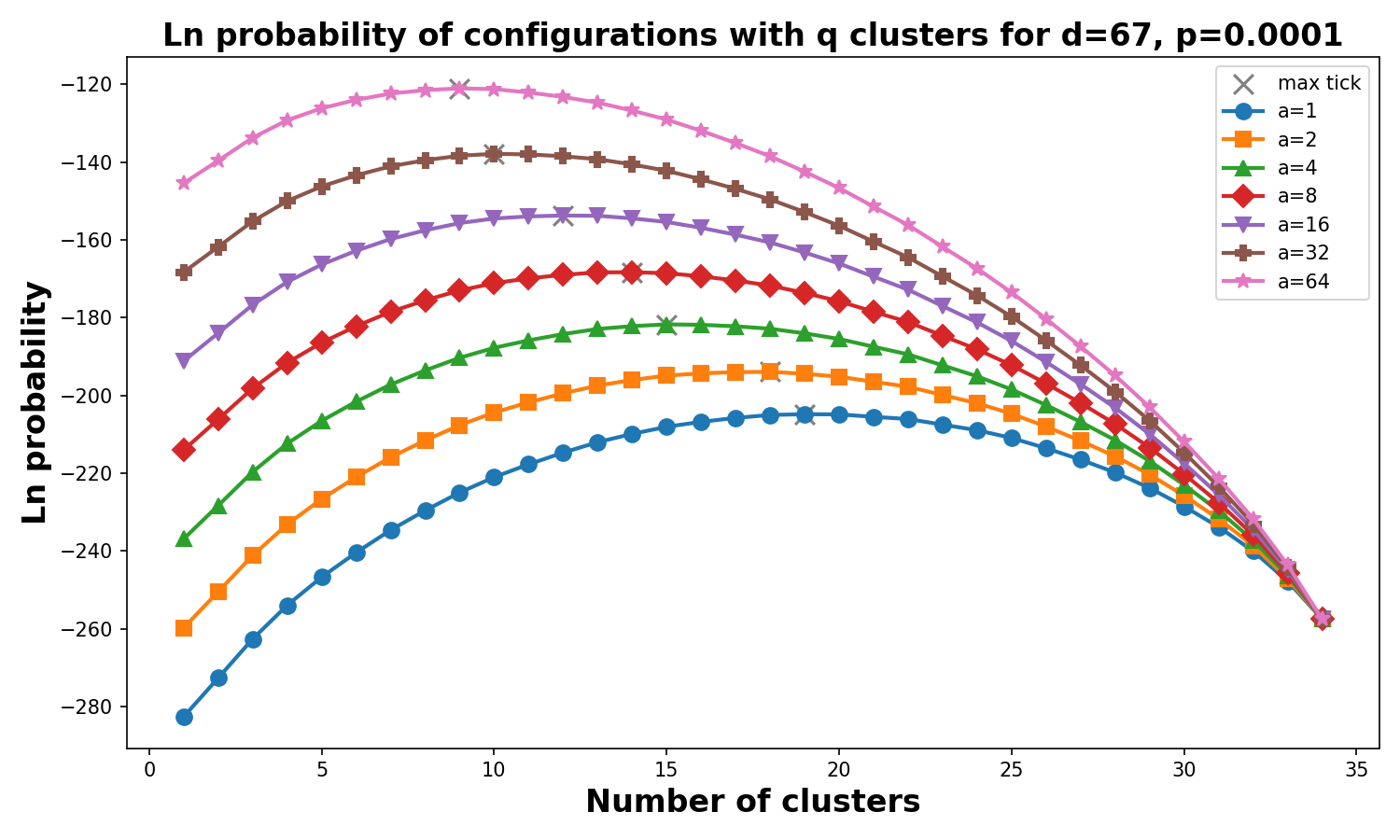}
  \caption{Ln probability of different configurations with $d_e$ errors for $d=67$, $p=10^{-4}$, and varying $a$. As $a$ increases, errors are more likely to cluster together.}
  \label{fig:configsa}
\end{figure}

\cref{fig:configsa} shows the likelihood of different configurations of error clusters with $d_e$ errors for varying $a$, $d=67$, and $p=10^{-4}$. The maximum starts at $19$ of $34$. This is slightly higher than the middle, because in a rotated code, spread out configurations of errors on MLLPs are slightly less restricted, and therefore, more likely. Increasing $a$ increases the average size of clusters, and therefore decreases the number of clusters. 
 \begin{figure}[thbp]  
\centering
\includegraphics[width=\columnwidth]{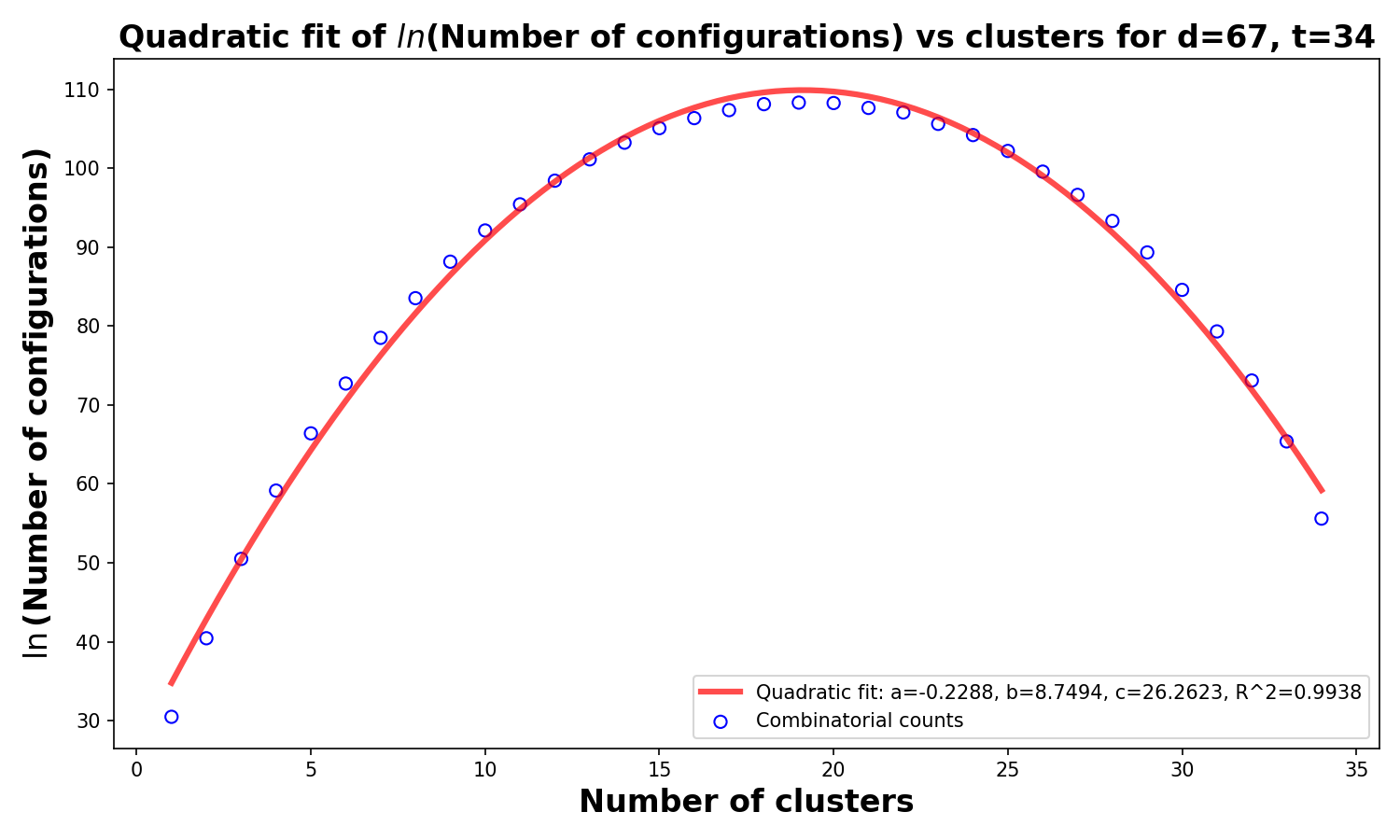}
  \caption{Quadratic fit for $\ln$(number of configurations) as a function of the number of clusters. The distribution is close to Gaussian. }
  \label{fig:configsfit}
\end{figure}

\cref{fig:configsfit} shows the number of configurations with with $d_e$ errors and $d=67$. The high $R^2$ confirms the quality of the Gaussian approximation, even for the rotated code.

\section{Globally correlated errors}
\label{sec:global}
We can extend the approximation of the logical error to a simple global correlated model. Suppose that instead of $p$ being drawn from a Bernoulli distribution, it is drawn from one of $D$ Bernoulli distributions, each occurring with a certain probability \textbf{for all qubits simultaneously}. Meaning

\begin{equation}
    \label{eq:correlatedP}
    p_{dep} = \sum_{j=1}^D \rho_j p_j, \quad \sum_{j=1}^D \rho_j = 1
\end{equation}
where $\rho_j$ is the probability of the $j$th distribution occurring, and $p_j$ is the physical error rate associated with that distribution. This type of distribution models errors occurring due to unfavorable environments such as those with fluctuating temperatures, manufacturing line defects, or in the presence of cosmic rays~\cite{li2025cosmic}. The form for the logical error is in this case: 
\begin{equation}
\label{eq:logicalErrorDependent}
    L_{dep} = \sum_{k=d_e}^{N} C_k \sum_{j=1}^D  \rho_j p_j^k (1-p_j)^{N-k}.
\end{equation}
Wlog, we can assume $\forall j$ $p_j>p_{j+1}$. Unless $\rho_1$ is very small, i.e. $\exists j$ where $\rho_1 p_1^{d_e}<\rho_j p_j^{d_e}$, we see that $j=1$ dominates the rest of \cref{eq:logicalErrorDependent}. In any case, there is some $d$ beyond which $j=1$ dominates the sum. Therefore, the scaling of any logical qubit error rate is dominated by $p_1$ at large enough code-distances, modulated by $\rho_j$. This shows that $L$ scales with the worst environmental physical error rate $p_1$ in \cref{eq:scaling}, but the constant $A$ is modulated by a factor $\rho_1$. Quantum computer designers should therefore consider worst-case, rather than mean or median, physical error rates when designing error correcting codes. 
A globally correlated error with $\rho_1=1/a$, $p_1=ap$, $\rho_2=1-1/a$, and $p_2=0$, has $p_{th}\approx (4a)^{-1}$. For large $a$, our globally correlated model has a threshold approximately $1/4$ of the locally correlated model in \cref{sec:corrUnrotated}.

\section{Measurement errors}
\label{sec:Meas}
Our analysis so far neglected measurement errors on the ancilla qubits for correction. We can repeat each calculation some odd $M\geq 3$ times independently, measure each time, and majority vote to decide the correct value of the ancilla qubit. We define $M_e\doteq (M+1)/2$. If we have $<M_e$ measurement errors on an ancilla qubit, we can trust it to perform as expected. Let $p_M$ be the probability that a given ancilla qubit has $\geq M_e$ measurement errors. We assume that a logical error due to measurement errors occurs if and only if at least one ancilla qubit has $\geq M_e$ measurement errors. This approximation is good when $p_Md^2\ll 1$, because measurement errors are unlikely to occur on multiple qubits in precisely a way that cancels out.

We follow the approach in \cite{Dennis_2002}, extending time to a third (upward Z) dimension of our surface code. A logical error occurs  when a majority of qubits in a tube parallel to the Z axis error. This behavior is identical to the other 2 dimensions in the unrotated code ~\cite{Dennis_2002}, but different than in the rotated code. Assuming the probability for a measurement error on any qubit is $p_m$ i.i.d, $M$ measurements are taken on each ancilla, and $N_a=\Omega(d^2)$ is the number ancilla qubits,
\begin{align}
    \label{eq:pMeasExact}
    &L_{M} = 1-\left[1-\sum_{k=M_e}^M \binom{M}{k}p_m^k(1-p_m)^{M-k}\right]^{N_a}\nonumber \\
    &\doteq1-\left[1-p_M\right]^{N_a}. 
\end{align}
If $p_m\ll 1$, $p_M\approx \binom{M}{M_e}p_m^{M_e}(1-p_m)^{M_e-1}$ because $\forall m\geq M_e, \; \binom{M}{k}\geq  \binom{M}{k+1}$. If $N_ap_M\ll 1$, $L_M \approx N_a\binom{M}{M_e}p_m^{M_e}$. Using \cref{eq:Stirling}, 
\begin{align*}
 L_M\approx \frac{N_a }{\sqrt{4M_e\pi}}(4p_m)^{M_e}\approx  \frac{M_e^{1.5}}{\sqrt{4\pi}}(4p_m)^{M_e}.   
\end{align*}

We want to select $M$ such that measurement errors are less frequent than data errors.  For an unrotated code, based on \cref{eq:Ccol}, $L_{ur}\approx  \sqrt{\frac{d_e}{\pi}}(4p)^{d_e}$. The exponential scaling in $d_e$ is identical to the corresponding scaling in $M_e$. 
For rotated code, based on \cref{fig:scalingfit} and our code, $L_r\approx A(p/p_{th})^{d_e}$ with $A\approx 2.09\tenp{-1}$ and $p_{th}\approx 7.33\tenp{-2}$. The exponential scaling in $d_e$ has a larger base than the corresponding scaling in $M_e$. 

For given $p, \;p_m,\; d$, we can calculate the necessary $M$ to ensure $L_M<L$. $M=d$, and $p_m=p$ is a particularly interesting case, because the dependence on $p,\; p_m$  vanishes. 
We compare the terms and conclude that $\forall p$ (that fulfills our other assumptions) $M=d$ is insufficient for $L_{ur}<L_M$, but sufficient for $L_{r}<L_M$.

\section{Summary and future work}
\label{sec:future}
We introduce a novel recursive algorithm for calculating logical error rates for rotated codes when $pd^2\ll 1$ by counting MLLPs. It calculates $C_{d_e}$ as a function of $d$, but runs in seconds even for large $d$. We can plug the result for any $p$ into \cref{eq:logicalError}. When $pd^2\gtrsim 1$ it provides a lower bound to the logical error. We further derive analytical lower and upper bounds to estimate error rates when $pd^2\gtrsim 1$. These bounds show that increasing the code distance is helpful, but less so than known scaling laws would predict when $pd^2\gtrsim 1$. The lower and upper bounds are tight when $pd<1$. 

We show how to expand our algorithm to account for a particular form of locally correlated errors. We show how to analytically accurately extrapolate the results of our algorithm to account for globally correlated errors.
We show how to extrapolate our analysis for a simple model with measurement errors.

Our method of counting configurations is only valid when the qubits are identical, but it is not necessary for them to be independent. Future work may consider more general qubit dependencies in space and time, more general measurement errors, and more efficient logical encodings such as quantum low-density parity check codes.
Our recursive counting method could be useful for combinatorial problems in other domains. 

\section*{Acknowledgments}
This research was funded by the DARPA Multi X Office's Quantum Benchmarking Initiative, contract number O2508-097-089-117256. The views, opinions and/or findings expressed are those of the authors and should not be interpreted as representing the official views or policies of DARPA or the U.S. Government. We thank James Nutaro, Andrea Delgado, Elaine Wong, Corey Fox, and Nina Gottschling of Oak Ridge National Laboratory for helping verify the methods and derivations in this manuscript.

\begin{appendices}
\section{}
\label{sec:AppendixA}
\begin{align}
 \label{eq:Ccol_proof}
     C_{d_e} = d\binom{d}{d_e}=\frac{d}{2}\binom{d+1}{d_e}\approx \frac{d/2}{\sqrt{\pi d_e}}2^{d+1}\lesssim  \sqrt{\frac{d_e}{\pi}}4^{d_e}
\end{align}

The simplifications in \cref{eq:Ccol_proof} come from 
\begin{equation}\label{eq:Stirling}
\begin{aligned}
   & \binom{2n}{n}=\frac{2n!}{n!n!}\approx \frac{\sqrt{2\pi 2n}\left(\frac{2n}{e}\right)^{2n}}{\left(\sqrt{2\pi n}\left(\frac{n}{e}\right)^{n}\right)^2}=\frac{1}{\sqrt{\pi n}}2^{2n}, 
\\ & \binom{2n-1}{n}=\frac{1}{2}\frac{2n(2n-1)!}{n!n(n-1)!}=\frac{1}{2}\frac{2n!}{n!n!}=\frac{1}{2}\binom{2n}{n}  
\end{aligned}
\end{equation}

\end{appendices}

\bibliographystyle{plain}
\bibliography{bibfile}
\end{document}